\documentclass[runningheads]{llncs}
\usepackage[latin1]{inputenc}
\usepackage{algorithm}
\usepackage[noend]{algpseudocode}
\usepackage{float}
\usepackage{graphicx}
\usepackage{mathtools}
\usepackage[title]{appendix}
\usepackage{soul, color}
\usepackage{comment}
\usepackage{thmtools,thm-restate}

\usepackage[colorinlistoftodos]{todonotes}

\begin{document}
\title{Towards Secure and Efficient Payment Channels}
%

\author{Georgia Avarikioti \and Felix Laufenberg \and Jakub Sliwinski \and Yuyi Wang \and Roger Wattenhofer}
\authorrunning{Avarikioti et al.}
%
\institute{ETH Zurich\\
\email{\{zetavar,felixla,sljakub,yuwang,wattenhofer\}@ethz.ch}}
\maketitle              
\begin{abstract}
Micropayment channels are the most prominent solution to the limitation on transaction throughput in current blockchain systems. However, in practice channels are risky because participants have to be online constantly to avoid fraud, and inefficient because participants have to open multiple channels and lock funds in them. To address the security issue, we propose a novel mechanism that involves watchtowers incentivized to watch the channels and reveal a fraud. 
Our protocol does not require participants to be online constantly watching the blockchain. The protocol is secure, incentive compatible and lightweight in communication. Furthermore, we present an adaptation of our protocol implementable on the Lightning protocol. Towards efficiency, we examine specific topological structures in the blockchain transaction graph and generalize the construction of channels to enable topologies better suited to specific real-world needs. In these cases, our construction reduces the required amount of signatures for a transaction and the total amount of locked funds in the system.   


\keywords{Payment Channels  \and Bitcoin \and Lightning \and Watchtowers \and Layer 2 \and Channel Factories}
\end{abstract}
\section{Introduction}
Increasing the transaction throughput of blockchain protocols without sacrificing security is arguably the biggest challenge for cryptocurrencies. 
So-called ``layer two'' protocols provide an elegant solution especially for payment transactions. Early versions \cite{spillman-contract} supported only unidirectional payments, but allowed these transactions to be confirmed instantly and off-chain, reducing the load of on-chain transactions.
Duplex channels \cite{DW2015channels} allow for bidirectional transactions between two parties. More recent work \cite{deckereltoo} has improved the efficiency and usability of channels.
Channels can be used to build a channel network using Hashed Timelock Contracts (HTLCs), as discussed in \cite{DW2015channels,poon2015lightning,malavolta2017concurrency,channelnetworks}. A payment channel network supports off-chain payments between users that do not share a direct channel. 
The Lightning network \cite{poon2015lightning} on Bitcoin \cite{nakamoto2008bitcoin} and the Raiden network \cite{raiden2017} on Ethereum \cite{ethereum} are implementations of such channel networks. Perun \cite{dziembowski2017perun} supposes a new network structure, building around payment hubs in order to make channel networks more efficient. The authors of \cite{dziembowski2017perun} also introduce a new model of channels that decrease the involvement of intermediaries for payments between parties that don't share a direct channel.

The security of channels relies on participants being responsive and running a full node. The main problem is that one party might try to settle on an old state of the channel, in which case the other party needs to publish a newer state during a specific time period. Thus, fraud can be proven and punished, as long as both parties are constantly online watching the blockchain. However, this is a major drawback for payment channels; being constantly online is not efficient and furthermore a risky assumption. One party could always launch a DDOS attack against the other party to prevent a new update transaction from being published in the blockchain. To address this problem, the  Bitcoin community has proposed using \emph{watchtowers} \cite{watchtowers} as third parties to watch channels, but to the best of our knowledge, incentivizing watchtowers remains an open problem. This is the first problem we examine in this work.

In a simple watchtower solution, miners could act as watchtowers to keep track of all channels off-chain and publish a proof of fraud to collect fees. Such an approach faces various problems. Firstly, throughput; having all miners receive, store and send all off-chain payments leads to network congestion which in turn leads to lower throughput and recentralization as more powerful hardware is required to become a miner. Secondly, it is difficult to predict the miners' behavior in such a system since miners are not incentivized to propagate the proofs of fraud. On the contrary, it might be in the miner's best interest to keep the proofs to himself to increase his probability to collect the fees if he publishes a block.
Our approach builds a network of watchtowers around each participant of the payment channel network. Watchtowers increase their expected payoff by faithfully executing the protocol, watching the blockchain and forwarding channel updates to other watchtowers. We prove that our protocol is secure and incentive-compatible. Furthermore, we present an adaptation of our protocol implementable on the Lightning network.

Although the construction of two-party payment channels is a simple and elegant solution, it fails to capture more complex topological structures in the blockchain transaction graph. Typical payment channels are funded by two people who lock their funds in a transaction in the blockchain to open the channel. For each update transaction both parties must sign to make it valid, thus two signatures are required. In an attempt to generalize the construction of channels and allow funds to be transferred freely between multiple parties, Burchert et al.~\cite{Burchert2017scalable} proposed the construction of channel factories. A channel factory is a multi-party channel. Channel factories reduce the number of open channels on the blockchain and also the total amount of locked funds since they allow money to be transfered off-chain between the participants of the channel factory. However, every valid transaction requires the signature of every party of the channel factory. This is highly non-practical. 

In this work, we examine specific topologies where multi-party channels do not require multiple signatures. Inspired by real life hierarchical situations, such as  customers paying a supermarket or citizens paying their taxes, we introduce a new type of channel construction that can capture such topologies. Our solution increases the efficiency of channels, when every node of the transaction graph has only one outgoing edge. Only two signatures are required for each transaction while the amount of locked funds is minimized. 

Our contribution is summarized as follows: In Section \ref{sec:protocol}, we present a novel protocol to secure channels by incentivizing  watchtowers. We prove our protocol is secure and incentive compatible in Section \ref{sec:analysis}, and present an alternative version implementable in Bitcoin's Lightning network in Section \ref{sec:bitcoin}. Additionally, in Section \ref{sec:newchannels} we  generalize the construction of channels to enable them to capture different topological structures of the blockchain transactions graph, thus improving efficiency. We discuss related work in Section \ref{sec:relatedwork} and conclude with future work in Section \ref{sec:conclusion}. The omitted proofs can be found in the appendices.


\section{The DCWC (Disclose Cascade Watch Commit) Protocol} 
\label{sec:protocol}

There are many different models and implementations of channels such as \cite{deckereltoo,DW2015channels}. While the implementations (e.g., Raiden, Lightning) differ, the concept remains mostly the same. 
A channel protocol typically consist of three phases:  setup, negotiation and settlement.

More specifically, two parties create on-chain funding transactions to start a channel, blocking parts of their stake on the blockchain. After this setup phase, both parties sign subsequent off-chain update transactions representing the current distribution of stake between them. Money can flow in either direction, as long as the paying party still has funds in the channel. This is the negotiation phase. Either party can create a settlement transaction to close the channel. The problem is if party  $A$ fraudulently issues an old settlement transaction, and party $B$ can prove it is old by presenting a transaction signed by $A$ with a higher sequence number. We have a proof-of-fraud, and party $B$ is awarded all stake in the channel.

Whenever an update transaction is issued, the party $B$ who receives a payment wants to be sure that some part of the network knows about this update. 
		A simple solution is to spread the newest update transaction to as many nodes as possible and offer a reward for whomever sends a proof of fraud to the miners in case of misuse. Such a protocol could accidentally or purposefully cause congestion at block miners, providing the possibility to launch cheap DDOS attacks. Nobody knows how many nodes have seen the update and are willing to send it to the miners.  Also nodes have no real incentive to spread the contract as it would only increase the competition, and decrease their chances of getting paid.
		  
		Our protocol requires participating watchtowers to store only $O(1)$ messages as the message with the highest sequence number is a proof for all others. In terms of \textit{privacy} watchtowers know which channel they are watching but they don't need to learn anything about the payments except for the sequence number. It consist of $l$ rounds. During each round a group of watchtowers gets the chance to send a proof-of-fraud and collect fees. Once a proof-of-fraud has been included on the blockchain or $l$ rounds have passed, the protocol terminates. Each round consists of some predetermined number of blocks $m$.  If no proof of fraud was published after $l$ rounds, the settlement transaction will be accepted as the final state of the channel and the funds will be unlocked. The protocol terminates after $l \cdot b < t$, where $t$ is the timelock off settlement transactions in current solutions \cite{poon2015lightning,DW2015channels}. We improve the time that funds are locked after a channel is closed one sided, compared to state-of-the-art solutions. Intuitively, our protocol can terminate much earlier since many more entities are online and watching the channel. Note that the channel can be closed instantly if both involved parties are cooperative and online. 
	
	When a new channel is opened, both parties provide a fund on-chain for paying transaction fees for a proof of fraud and paying the watchtower that published it. The protocol consist of two phases.  
	 The first phase is executed once for every newly issued update transaction. It is responsible for spreading the update in the network. 
	 The second phase ensures that a proof of fraud will be published, when an old update transaction is published as a settlement transaction. The second phase consists of $l$ rounds, where $l$ is specified on the funding transaction and should depend on the amount of funds locked in the channel.

	\subsection{Phase 1: Disclose \& Cascade}
    The first phase is entered after the funding transaction has been included in the blockchain and terminates when a settlement transaction is included in a block. This phase ensures that a sufficient amount of watchtowers receive an update message, and also that these messages are correct. 
    Essentially, participants send new update messages to their neighbours for a few iterations. These messages are cryptographically designed to capture the travelled path in order to reward watchtowers that forwarded messages.
	\textsc{Disclose} of Phase 1 is invoked by a party $A/B$, involved in the channel, anytime he or she receives a new update transaction $t_{c,i}$ with sequence number $i$. \textsc{Cascade} of Phase 1 is invoked by any node receiving a message generated from \textsc{Disclose} or \textsc{Cascade}. The parameter $N$ limits the number of neighbors to whom a message can be sent. The value $d_{m}$ determines the number of hops that a message $m$ has traveled. Let $\{m\}_K$ denote that $m$ is encrypted with private key $K$.
 
	\begin{algorithm}[t]
    	\label{alg:A}
		\caption{Phase 1: \textsc{Disclose \& Cascade}}
		\begin{algorithmic}[1]
			\Procedure{Disclose}{}
            \label{proc:init}
			\State $\textit{K} \gets {own~private~key}$

			\For{($i \gets 1$;~$i \le N$;~$i\gets i+1$)} 
				\State $W \gets find~new~neighbour$ 
				\State $m \gets \{i,W,t_{c,i}\}_K$
				\State send($m$) to $W$
			\EndFor

			\EndProcedure
			
			\Procedure{Cascade}{}
           \label{proc:spread}
			\State $\textit{in} \gets {receive~message}$
			\State $\textit{t} \gets {d_{in} + 1}$
			\If{$l \le t$} \Return
			\EndIf
			\State $\textit{K} \gets {own~private~key}$

			\For{($i \gets 1$;~$i \le N$;~$i\gets i+1$)} 
			\State $W \gets find~new~neighbour$ 
			\State $m \gets \{i,W,in\}_K$
			\State send($m$) to $W$
			\EndFor
			
			\EndProcedure
		\end{algorithmic}
	\end{algorithm}
	
Thus,

\begin{equation}
    \begin{array}{lll}
	t_{c,i} & \coloneqq &  \text{i-th update transaction of channel c} \\ \\
	
	m & \coloneqq & \{id,W,m'\}_K~or~ \{id,W,t_{c,i}\}_K
	\end{array} 
\end{equation}


\subsection{Phase 2: Watch \& Commit}
	The second phase of the protocol ensures that in case of fraud, i.e. $A$ or $B$ publish an old update transaction, a proof of fraud will be published. Throughout the lifespan of the channel all watchtowers that are involved in the protocol watch the blockchain for such an update message (\textsc{watch} process).
	After a settlement transaction $t_{c,i}$ of a channel, signed by one of the parties involved in the channel, is published in a block the watchtowers start the \textsc{commit} process. It consists of a fixed number of rounds $l$, specified in the channel. For simplicity, we assume a round corresponds to a single block. However, the number of blocks that determine a round can be adapted without affecting the protocol. We note that block frequency and propagation time should be considered to estimate the reaction time of the watchtowers, and thus the number of blocks per round.
    The actors in this phase are the watchtowers. 
    
    Let $t_{c,j}$ denote the newest update transaction that watchtower $W$ has seen. Let $m$ denote the message as created in the \textsc{Disclose \& Cascade} phase through which $W$ learned about the newest update transaction $t_{c,j}$. Watchtower $W$ stores only this message and can discard any other message concerning channel $c$ thus reducing his storage costs. After $d_{m}$ blocks trail the settlement transaction $t_{c,i}$ and $j > i$ and no proof of fraud has been published yet, $W$ signs and sends $m$ to the blockchain network. With $m$, $W$ also sends all information that $m$ contains in plaintext. This information includes the path that the message has travelled, and the $id$ at each hop of the path. The plaintext information will not be included on the blockchain, but helps the network to reduce the number of invalid messages that are send around. If the plaintext of the message does not match the encrypted information, then the network nodes will discard the message.
    

	\subsection{Payoffs}
  After $l$ rounds have passed, the channel is closed. At this stage the funds of $A$ and $B$ are unlocked and the watchtowers are paid.
  
  Ideally, all watchtowers would get pay their marginal contribution for participating in the protocol.  However, a protocol that implements such a payment mechanism would require many on-chain payments to pay all involved watchtowers. Thus, we propose a payment mechanism that pays only the watchtowers that participated in forwarding and publishing the update message that has been included on-chain. 
  When the channel is created, both parties involved in the channel reserve some value $\rho_c$ for paying transaction fees and rewards for a proof-of-fraud. When a proof-of-fraud has been included in the blockchain, all watchtowers included in forwarding that proof-of-fraud get an equal share of $\rho_c$, reserved by the cheating party, after the transaction fee has been deducted.
  In the following section, we show that this payment mechanism ensures sufficient expected payoff to incentivize the watchtowers.
  
\subsection{Transaction validation}

The miner of the next block checks whether the messages were generated according to the protocol. Out of the set of valid messages he or she selects one uniformly at random and publishes it in the new block.
The miners perform the following checks to determine the correctness of a message $m$ containing the update transaction $t_{c,j}$:
	\begin{itemize}
    	\item No proof of fraud for transaction $t_{c,j}$ has been published.
		\item The update transaction $t_{c,j}$ was generated according to underlying channel protocol.
		\item $j > i$.
		\item  $1\leq id \leq N$.
		\item No other message in the same level following the same path has the same $id$.
		\item The level of the watchtower that propagates the message, determined through the number of signatures, is equal to the number of rounds passed since the settlement transaction was included in the blockchain. 
		\item For every message in the same path as $m$: the $pk$ of the previous message corresponds to the secret key $sk$ signing the current message.
		\end{itemize}
	If any of the points above are violated the network will reject $m$.
	
	  \subsubsection{Payoff for Miners.} 
  The miner including the settlement transaction is paid as in any other transaction. The miner including a valid proof of fraud is paid by the watchtower whose message he included. The reward for the watchtower should thus be high enough to cover the transaction fees.

\section{Incentive \& Security Analysis}
\label{sec:analysis}

In this section we show that watchtowers maximize their profit if they follow the \textsc{DCWC} protocol. Hence, executing the protocol is a dominant strategy and thus the protocol is incentive compatible.

\subsection{Expected Payoff}
In order to show incentive compatibility we need to calculate the expected payoff and analyze changes in the expected payoff when a node decides to deviate from the protocol. 
The expected payoff of watchtower $W$ watching channel $c$ is 


\begin{equation}
\label{eq:payoff_new}
\begin{array}{ll}
E[pay_{W,c}] & =\sum_{i = 1}^l P[S_{W,i}] \cdot \frac{\rho_{c}}{i}\\ \\
& =\sum_{i = 1}^l \sum_{m \in S_{W,i}} P[m] \cdot \frac{\rho_{c}}{i}\\ \\
& = \sum_{i = 1}^l |S_{W,i}| \cdot P[m] \cdot \frac{\rho_{c}}{i}
\end{array}
\end{equation}

where $S_{W,i}$ is the set of messages signed by $W$ such that $d_m = i$, including those messages forwarded by $W$. Then $P[S_{W,i}]$ denotes the probability that one of those messages will be included in the blockchain. The equally shared payoff $\rho_{c}$ is split between all watchtowers that forwarded $m$, paid if $m$ gets included on the blockchain. $P[m]$ is the probability that update message $m$ will eventually be included in the blockchain. To calculate this probability for a message sent in the \textsc{Disclose} step of Phase 1, we sum the probabilities of $m$ being chosen for all possible subsets of non-failing watchtowers. For messages generated in during  \textsc{Cascade} to be included in the blockchain, all watchtowers holding messages that travelled fewer hops must have failed. If we assume all nodes follow the protocol and nodes fail independently with probability $\alpha$, then the probability that message $m$ will be included in the blockchain is:

\begin{equation}
\label{eq:probability}
\begin{split}
	P[m] &  = \left\{\begin{array}{lr}
	\sum_{i=1}^{N} \frac{1}{N} (1-\alpha)^{i} (\alpha)^{N-i} \binom{N}{i}, & \text{for } d_{m} = 1\\ \\
	\sum_{i=1}^{|S_i|} (\frac{1}{|S_i|} (1-\alpha)^{i} (\alpha)^{|S_i|-i}\binom{|S_i|}{i}) \cdot \alpha^{\sum_{j = 1}^{i-1} |S_j|},~ & \text{for } 1 < d_{m} \le l\\ \\
	0 , & \text{for } d_{m} > l
	\end{array}\right\}
\end{split}
\end{equation}


Where $|S_i|$ is the total number of valid send messages with depth i. If the protocol is executed faithfully by all parties, then $|S_i| = N^i$.

\begin{lemma}
\label{lem:1_0}
A watchtower cannot increase his expected payoff $E[pay_{W,c}]$ by deviating from the protocol unless he can decrease $d_{m}$ for any of his messages, increase the probability $P[m]$ that one of his messages will be included by decreasing the number $|S_i/S_{W,i}|$ of valid messages not signed by him or increase the number $|S_{W,i}|$ of messages signed by him for any $i$.
\end{lemma}

\begin{proof}
Lemma \ref{lem:1_0} follows directly from equations \ref{eq:payoff_new} and \ref{eq:probability}.
\end{proof}

 Lemmas \ref{lem:1_1}, \ref{lem:1_2}, \ref{lem:1_3} complete our argument by showing that no watchtower can decrease $d_{m}$ or increase $P[m_W]$ or $P[S_W \setminus m_W]$. We prove this for Phase 1 and Phase 2 of the protocol independently.

\subsection{Analysis of Phase 1: Generating Messages}

The first phase of DCWC handles the creation of update messages and how they are forwarded to the watchtowers. We show that during the first phase of DCWC, watchtower $W$ can not influence the execution of the protocol in a way that decreases the number of messages that do not pay him nor in a way that increases the number of valid messages that do pay him.

\begin{restatable}{lemma}{lemmaa}
\label{lem:1_1}
A watchtower $W$ can not decrease $|S_i/S_{W,i}|$ nor increase $|S_{W,i}|$ by deviating from executing Phase 1 of DCWC.
\end{restatable}

 Furthermore we show that he can not manipulate the number of signatures on any of his messages.

 \begin{restatable}{lemma}{lemmab}
 \label{lem:1_2}
A watchtower $W$ can not Decrease $d_{m}|m \in S_{W,i} ~ \forall i$ by deviating from executing Phase 1 of DCWC.
 \end{restatable}

 \subsection{Analysis of Phase 2: Committing Messages}
 The second phase of the protocol handles the spreading of a valid proof-of-fraud through a update message $m$. Watchtowers might want to withhold proof of frauds, publish them early or deviate in other ways to increase their expected payoff. 
 
 Note that there is no need to show that $W$ is not able to decrease $d_m$ in this phase, as we assume that Phase 1 has been completed and all update messages have been created accordingly.
 It is also easy to see that $W$ can't decrease  $|S_i/S_{W,i}|$ as we assume that Phase 1 has already completed and $W$ has no control over the choices that the owners of $|S_i/S_{W,i}|$ take.

  \begin{restatable}{lemma}{lemmac}
  \label{lem:1_3}
  A watchtower $W$ can not increase $|S_{W,i}|$ by deviating from executing Phase 2 of DCWC.
  \end{restatable}

\begin{theorem}
\label{Th:1}
DCWC is incentive compatible.
\end{theorem}
\begin{proof}
Theorem \ref{Th:1} follows directly from Lemmas \ref{lem:1_0}, \ref{lem:1_1}, \ref{lem:1_2}, \ref{lem:1_3}.
\end{proof}

  \subsubsection{Mining Nodes}
Mining nodes might also deviate from the protocol to improve their payoffs. Their set of actions is a bit more limited as any block containing invalid update messages would be rejected by other nodes.
	The analysis thus far has assumed that the set of watchtowers is disjoint from the set of mining nodes. A node $W$ which is mining a new block and is holding a valid update message, can give himself an advantage of receiving the payoff by not randomly including a proof-of-fraud but rather including his own. This behavior does not influence the outcome of the protocol as it leads to a proof-of-fraud being published either way. The mining node can only include a valid proof-of-fraud in his block without invalidating the block. Thus mining nodes have an additional incentive to participate as watchtowers in our protocol.

\section{Adaption for Compatibility with Lightning: The DCWC* Protocol}
\label{sec:bitcoin}

So far, we described the protocol to incentivize watchtowers to watch the channels in high level. An implementation of the protocol depends on the underlying channel protocol and the underlying blockchain protocol. An implementation in Turing complete blockchains can be done as is. An implementation for blockchain protocols that support more limited scripting languages, such as Bitcoins' Script, would require changes to the protocol. Most prominently limiting the number of forwarded messages at each step is, to the best of our knowledge, not possible to be implemented for Lightning due to the limitations of Script. Furthermore, determining how to pay the watchtowers without explicitly naming them in the funding or update messages is not trivial. The protocol depends not only on the scripting language, but also on the channel implementation. We propose a simplified version of our protocol that leaves unspent transaction outputs to be claimed by watchtowers, designed for the current Lightning implementation \cite{poon2015lightning}. DCWC* can be implemented without requiring changes to the channel protocol.

Just as in DCWC, DCWC* creates layers of watchtowers. The $i$-th layer of watchtowers is allowed to issue a proof of fraud after $i$ rounds trail the occurrence of a fraudulent settlement transaction, also similar to the DCWC protocol. The protocol terminates after a predetermined number of blocks trail the settlement transaction, making funds and rewards spendable. 
Rewards are granted for the watchtower that submitted the published proof of fraud to the miners and the watchtowers that forwarded that proof of fraud.

The first phase, \textsc{disclose \& cascade*}, is executed once for every off-chain update of the channel. This phase is responsible for spreading the update in the network of watchtowers.

\subsection{Phase 1: Disclose \& Cascade*}
This phase initiates after the channel is opened and ends when a settlement transaction has been published in the blockchain. \textsc{Spread} ensures that a certain amount of watchtowers receive update messages, and that the messages are constructed such that watchtowers receive a payoff for storing and forwarding the messages.

Whenever $B$ receives an update transaction from $A$, he also receives a transaction that invalidates the previous transaction, rewarding $B$ with all of the channels funds, if the previous state is published on-chain by $A$. We refer to this transaction, invalidating update transaction $t_{c,i}$ as \textit{invalidation transaction} $\tilde{t}_{c,i}$, which corresponds to a proof-of-fraud in DCWC. $B$ wants to be certain that some watchtowers are aware of this invalidation transaction. Thus, $B$ sends$\tilde{t}_{c,i}$ to watchtowers, claiming most of the funds for himself but leaving some unspent. When a watchtower $W$ receives such a message, he claims some of the unspent transaction outputs (UTXOs) for himself but leaves some unspent and forwards the message. $W$ timelocks the transaction which he added to the message. This way, he will always get a chance to publish a proof of fraud before whomever he sends it to and thus has no incentive not to forward the update. Note that if $W$ later gets the chance to publish a proof of fraud he can still claim all UTXOs for himself. Only if the recipient of his message publishes the proof, $W$ receives the part he claimed before forwarding the message. If the unclaimed UTXO is too small, the receiver will possible discard the message; thus a market of self regulation evolves to determine the appropriate amount for the watchtowers' fee. 

The construction of these update messages is depicted in Figure \ref{fig:construction}. If $A$ and $B$ agree on a new update $t_{c,i}$ of the channel they provide each other with the invalidation transaction $\tilde{t}_{c,i-1}$ of the previous transaction. Then each party can create one transaction, giving himself all of the funds in the channel and another transaction with a transaction-level relative timelock (\textit{nSequence}) to all his neighbours, leaving them in control of parts of the output. This can be repeated recursively by his neighbours until the sum of relative timelocks exceeds the timelock of the channel, or until it is reasonable to assume that adding another level makes the participation of more watchtowers unprofitable. 

	\begin{figure}[h]
		\centering
		\includegraphics[width=0.8\textwidth]{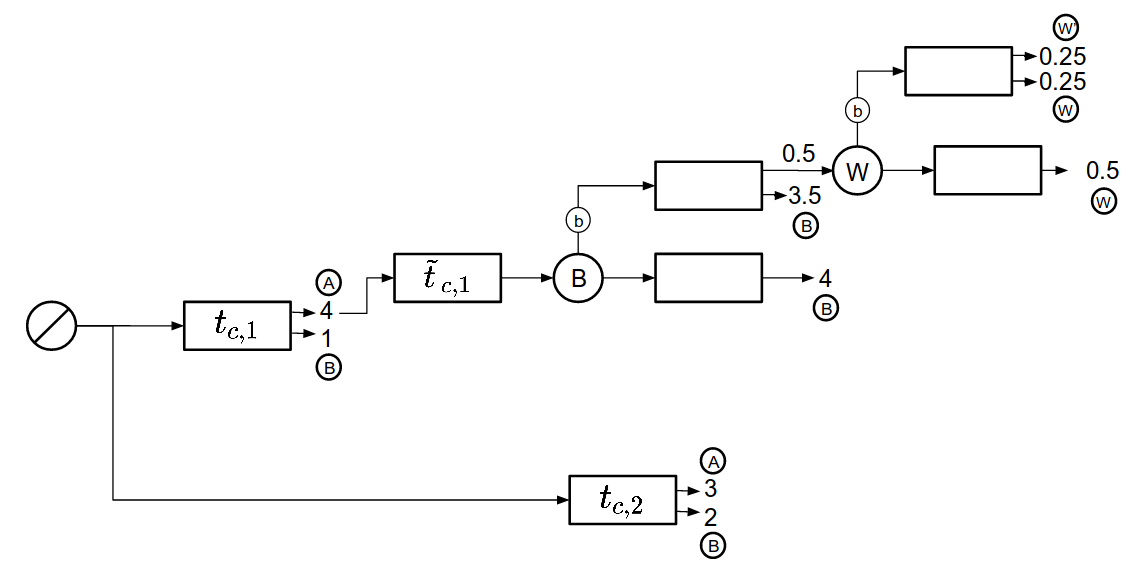}
		\caption{Participant $B$ in the channel \textsc{disclosing} the newest invalidation message $\tilde{t}_{c,1}$ to Watchtowers $W$ after $A$ and $B$ agreed on a new update $t_{c,2}$. $W$ \textsc{cascades} the message to $W'$ with a transaction level relative timelock $b$.} 
	    \label{fig:construction}
	\end{figure}
    
\subsection{Phase 2: Watch \& Commit*}
	The second phase of the protocol ensures that the fraud will be detected and proven, if it occurs.
	This phase initiates after the a settlement transaction $t_{c,i}$ is published on the blockchain by one of the involved parties. Since all $t_{c,i}$ are timelocked, the \textsc{commit} phase consists of a fixed number of rounds. Any watchtower holding $\tilde{t}_{c,i}$ can now send it to the blockchain network. Then they can, round after round, append the transactions leading to their payoff.

\iftrue
\subsection{Implementation}
Our protocol has only one type of no-trivial transaction. Namely, the transaction that is forwarded to a watchtower (depicted as a circle in fig. \ref{fig:construction}).It consists of two outputs, one going to the sender of the message $B$ and one that is locked by the script shown in figure \ref{fig:op_code}. The forwarded transaction has the purpose of paying the sender of the transaction and the watchtower receiving it, if the sender is offline when the message is included on-chain. The output script is constructed as follows:

\begin{figure}
\hrulefill \\
OP\_IF\\
\hspace*{0.6 cm} $b$ OP\_CHECKSEQUENCEVERIFY \\
\hspace*{0.6 cm} \textless W's public key\textgreater \\
\hspace*{0.6 cm} OP\_CHECKSIG \\
OP\_ELSE \\
\hspace*{0.6 cm} \textless B's public key\textgreater \\
\hspace*{0.6 cm} OP\_CHECKSIG \\
\hrule
    \caption{Caption}
    \label{fig:op_code}
\end{figure}

 Recall that the funding transaction, $t_{c,i}$ and $\tilde{t}_{c,i}$ are part of the channel protocol and not further discussed here.

\subsection{Discussion of DCWC*}

	\begin{itemize}
		\item \textit{Abolishing rounds.} In a practical setting it would make sense to abolish rounds, and with that the \textsc{cascade} process all together. This would lead to a higher workload for the participants of the channel as they have to create more messages, but reduces the number of transactions that have to be included on-chain.
        \item \textit{eltoo.}It remains to be seen how a version of DCWC* would work with eltoo \cite{deckereltoo} channels. While the DCWC protocol requires watchtowers to store only $O(1)$ messages, DCWC* requires watchtowers to store $O(n)$ messages, where $n$ is the number of channel updates. Possibly the proposed SIGHASH\_NOINPUT could be used to improve the implementation of DCWC* to require storing only $O(1)$ messages.
        \item \textit{privacy.} In DCWC* watchtowers receive a transaction $tx$, which does not leak more information than the id of the input transaction of $tx$. Thus, the watchtowers do not learn the distribution of funds until the settlement transaction is included on-chain.
	\end{itemize}
\fi

\section{Generalizing Payment Channels}\label{sec:newchannels}

In this section we study {\it extended domain channels}, or {\it xD-channels} for short, a protocol permitting new types of payment channels. The xD-channels can be custom-designed, depending on intended use, to optimize their functionality, such as the amount of locked funds or the number of signatures required to authorize a transaction.


The premise of traditional two-party channels is that payments can be executed off-chain in a restricted form to avoid double-spending; funds are locked to the channel in the funding transaction, and thus channel participants can only pay each other. We recognize that bi-directional, two-party functionality is rarely optimal, and relax the protocol by allowing each party to freely choose whom to pay within the channel. Note that double spending is not possible as long as each channel participant specifies upfront only one party eligible to be paid by the participant. Therefore, guarantees similar to those of traditional two-party channels are preserved.

\subsection{xD-channels}

\begin{definition}[xD-channel graph]
An xD-channel is described by a directed graph $G = (V, E)$, where $V$ is a user (a set of public keys), and $E$ is a directed edge such that each vertex has at most one outgoing edge.
\end{definition}

\begin{definition}[xD-channel state]
The xD-channel state is a tuple $(G, {\cal I}, k, {\cal S}, {\cal P})$, where:
\begin{itemize}
\item $G$ is the xD-channel graph.
\item ${\cal I} = \{(v_0, f_0),\dots,(v_n, f_n)\}$ is the initial assignment of funds $f_i$ to public keys $v_i$ in each constituent xD-channel.
\item $k$ is a natural number.
\item ${\cal S}$ are signatures of all public keys specified in ${\cal I}$, signing $(G, {\cal I}, k)$.
\item ${\cal P}$ is a set that for every vertex $v$ of $G$ contains the amount of funds $f$ that $v$ has paid along its' outgoing edge, together with a signature $s$ authorizing that payment, $(v,f,s)$.
\end{itemize}
\end{definition}

\noindent
An xD-channel is established by a funding transaction of the amount $\sum_{(v_i,f_i) \in {\cal I}} f_i$ to the channel $(\emptyset, {\cal I}, 0, {\cal S}, \emptyset)$.

Similarly to traditional payment channels, an xD-channel is closed by publishing its state. Given some xD-channel states published during the settlement period, only states with the maximal sequence number $k$ have effect. For $i=1,\dots,m$, let $(G, {\cal I}, k, {\cal S}, {\cal P}_i)$ be such states. Then, the effect of publishing those states is equivalent to publishing only $(G, {\cal I}, k, {\cal S}, {\cal P'})$, where ${\cal P'}$ contains for each vertex $v$ the maximal amount $v$ has signed as paid: $(v,f,s) : f = \max_{i=1,\dots,m} \{f' : (v,f',s') \in {\cal P}_i\}$. 

Suppose only one channel state $((V,E), {\cal I}, k, {\cal S}, {\cal P})$ is published during the settlement period. Then, for each vertex $v$ in $G$, the number $f_v$ of funds paid along $v$'s outgoing edge in the resulting state is determined by the solution to the following linear program:
\begin{align*}
\text{Maximize } \sum_{v \in V} f_v \text{ subject to: } & \forall_{(v,f,s) \in {\cal P}} f_v \le f \\
& \forall_{(v, i_0) \in {\cal I}} \sum_{(u,v) \in E} f_u + i_0 \ge f_v
\end{align*}

Note that the linear program determines the highest amount $f_v$ that $v$ has agreed to pay that does not result with a negative balance for $v$.

Participants ensure security of transactions in a similar, but generalized way to traditional two-party channels. Whenever a participant $u$ wants to transact funds to another participant along an edge $(u,v)$, they send $v$ an updated element $(u, f, s)$ of ${\cal P}$, where $f$ is increased by the funds paid, together with an authorizing signature $s$. $u$ proves to $v$ that the current state of the channel assigns enough funds to them to execute the transaction, by presenting $v$ with commitments by other participants to pay $u$: $\{(w, f_w, s_w) : (w, u) \in E\}$, such that $i_0 + \sum_{\{(w, f_w, s_w) : (w, u) \in E\}} f_w \ge f$ (where $(u, i_0) \in {\cal I}$), along with commitments funding those payments and so on. In turn, $v$ can pass over these elements of ${\cal P}$ to prove its' ability to pay, and so on. Crucially, because each vertex has at most one outgoing edge, there can be no attempt to double spend within the channel. Moreover, changes that are not merely commitments to pay more along participant's outgoing edge, such as changing the channel's topology $G$, require the signature of everybody involved.

Note that when the channel participants decide to close the channel, they might issue a final channel state $(\emptyset, {\cal I}, k_{\max}, {\cal S}, \emptyset)$ where the resulting fund division is concisely described by ${\cal I}$, to minimize the blockchain space required to publish it.

\begin{corollary}
A traditional, bi-directional two-party channel is equivalent to an xD-channel with $G = (\{A,B\},\{(A,B),(B,A)\})$.
\end{corollary}

\begin{example}
Consider a scenario in which many parties $c_0,\dots,c_m$, the {\it clients}, are interested in periodically paying a party $s$, the {\it supermarket}. The supermarket expects to never need to transfer funds to the clients, but there is another party $t$, the {\it tax office}, to which the supermarket wants to periodically make payments. The tax office expects to pay $c_m$ often.

One solution involving traditional two-party channels might be to establish a channel between each client and the supermarket and channels between the supermarket, tax office and the client. Note, that funds in each channel are locked to either party in the channel, e.g. the funds paid by the clients to the supermarket cannot be passed on to the tax office. Note, that each client-supermarket channel allows payments from the supermarket to the client, but this functionality is superfluous.

Consider a single xD-channel, where $G$ is illustrated in Figure \ref{figex1}. Note that clients can make payments to the supermarket identically as before. However, the supermarket can pass the funds on to the tax office providing the proofs of clients' payments along with its signature.
\end{example}

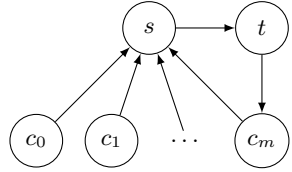
\begin{figure}
\centering
\begin{tikzpicture}

\node(1)[draw,circle,minimum size = 0.7cm] at (0,0) {$c_0$};
\node(2)[draw,circle,minimum size = 0.7cm] at (1,0) {$c_1$};
\node(3) at (2,0) {$\dots$};
\node(4)[draw,circle,minimum size = 0.7cm] at (3,0) {$c_m$};

\node(5)[draw,circle,minimum size = 0.7cm] at (1.5,1.5) {$s$};

\node(6)[draw,circle,minimum size = 0.7cm] at (3,1.5) {$t$};


\draw [-latex] (1)  edge (5);
\draw [-latex] (2)  edge (5);
\draw [-latex] (3)  edge (5);
\draw [-latex] (4)  edge (5);
\draw [-latex] (5)  edge (6);
\draw [-latex] (6)  edge (4);

\end{tikzpicture}

\caption{xD-channel from Example 1. Vertices represent actors, and edges represent the ability to pay.}
\label{figex1}
\end{figure}

\subsection{Discussion}


Similarly to traditional two-party channels, xD-channels allow payments to be executed off-chain while preserving the security guarantees of the blockchain, i.e. the parties cannot double spend. To that end, multi-party channels typically require the signature of every party involved. Towards this direction, we can extend the functionality of xD-channels even further.
So far, we require each xD-channel participant $i$ to choose only one other participant $j$ as the recipient of payments, so that the set of parties ${\cal S}$ reachable from $i$ in the xD-channel graph is protected against potential double spending from $i$. We note that $i$ can be allowed to make a payment to a party $k \neq j$ as long as each party in ${\cal S}$ confirm, by providing a signature, they are not being cheated by such payment. In the presence of channel topology, this approach can be more efficient, i.e. provide a better ratio of functionality to the number of required signatures, than typical multi-party channel approaches.


\section{Related Work}\label{sec:relatedwork}

Off-chain payment channels have been extensively studied by the research community as they are the most prominent solution to the blockchain's scalability problem. Multiple versions of channels can be found in literature . Duplex micropayment channels \cite{DW2015channels} use timelocks, while lightning channels \cite{poon2015lightning} depend on punishing the party that misbehaves. 

Payment networks can be build using any version of payment channels. Most payment networks use Hashed Timelock Contracts (HTLC) to execute transactions over multiple hops, such as the Lightning network \cite{poon2015lightning} which relies on Bitcoin \cite{nakamoto2008bitcoin}, and the Raiden network \cite{raiden2017}  which relies on Ethereum \cite{ethereum}. 
However, there are different approaches on constructing payment channels and building payment networks. In a recent work, Miller et al.~present Sprites \cite{Miller2017sprites} to reduce the time for which the funds are locked in a multi-hop transaction. Sprites also supports partial withdraws and deposits without interrupting the channels' functionality. 
On the other hand, Perun \cite{dziembowski2017perun} introduces ``virtual payment channels'', which are payment channels on top of the payment channels, to abolish intermediaries in multi-hop transactions in the payment network.
Our work is complementary to all these channel construction and payment networks, since we address a problem they all share:  all parties involved in channels must be constantly online to ensure security.

A concurrent work, addresses the same problem. McCorry et al.~propose Pisa \cite{mccorry2018pisa}, a protocol that introduces third parties, called custodians, to watch the channels. Their approach focuses on state channels; they set up service contracts between channel parties and custodians. They also address the incentives for participation and propose custodians depositing security funds when setting up the service contract. Although our protocol does not examine this issue, it is very simple and easily implementable, even on Bictoin, in contrast with the Pisa protocol.


\section{Conclusion \& Future Work}\label{sec:conclusion}
In this work, we presented a mechanism to secure channels by incentivizing third parties, called watchtowers, to actively monitor channels and report fraud to the blockchain. The mechanism is incentive compatible, i.e. following the protocol is a dominant strategy for every watchtower, and allows the parties involved in the channel to go offline. The proposed protocol is lightweight in communication and watchtowers do not learn the distribution of funds in the channel. In addition, we suggested an adaptation of the protocol implementable on the Bitcoin's Lightning network.

Furthermore, we explored channels efficiency. We generalized the construction of channels to allow specific topological structures of the blockchain transaction graph to influence the multi-channel construction. This way, we improve the efficiency  of channels in some specific cases, which are often met in practice in monetary systems. In particular, when each party has only one outgoing edge, each transaction in the multi-party channel requires only one signature. Moreover, the proposed construction enables transfering the money from one channel to the next, similarly to how IOUs work, thus reducing the amount of locked funds in the channel.

For future work, an interesting open topic is to study participation incentives for the watchtowers. Although our protocol is incentive compatible in the sense that if watchtowers decide to participate they cannot gain more by deviating from the protocol, we do not guarantee that they actually profit by participating in the protocol, particularly when no fraud occurs. 
Another line of future work would be to examine the new channels construction. The first open question would be whether parties can be added or removed efficiently without going to the blockchain. Another direction would be to study different topologies for which the construction greatly improves channels' efficiency, either regarding the amount of locked funds or the number of required signatures.





\newpage
\bibliographystyle{splncs04}
\bibliography{references}
%
%
%
%
\newpage
\appendix
\section{Proof of Lemma \ref{lem:1_1}}

\lemmaa*

\begin{proof}
Picture the \textsc{Disclose \& Cascade} algorithm as a tree. 
$W$ can increase $E[pay_{W,c}]$ by increasing the number of valid messages with his signature. This is true because, whenever a message $m$ from a node in his subtree is published $pay_{W,c} > 0$ and in all other cases $pay_{W,c} = 0$. Note that maximizing the number of  messages  $|S_{W,i}|$ in any layer bears no effect on the probability of messages of earlier levels to be included; thus $W$ doesn't decrease his chances of a larger payoff by maximizing the number of valid messages in any level of his subtree. Therefor we can analyze each level individually. \newline \newline


It still remains to be shown that $W$ can neither  decrease $|S_i/S_{W,i}|$ nor increase $|S_{W,i}|$. Let $M$ denote the set of all update messages. The protocol could only be exploited in one of two ways: \\

 \textbf{1. Decrease $|S_i/S_{W,i}|$:} Decreasing the number of update messages that are not in his subtree would increase  the chance that one of $W$s messages gets included.
 
 	$W$ has no knowledge of the identities of watchtowers that are not in his subtree. Additionally their success does not depend on any action of $W$. Thus, decreasing $|S_i/S_{W,i}|$ is not feasible for $W$.
 
 \textbf{2. Increase $|S_{W,i}|$:} We show that $W$ can't increase the number of nodes in his subtree. 

 For each depth the number of messages that can be send is restricted by $N$. If $W$ tries to forge more messages, then according to the pigeonhole principle either two $ids$ must be the same, or at least one $id$ must be larger than $N$. 
  In the second case, all descendant messages will be rejected. 
 The first case is a bit trickier. Sending two different messages with the same sequence number to disjoint set of miners is clearly not more beneficial than just sending one message to all. However if the failure probability of nodes is high, creating two messages with the same $id$ increases the chances of one of them to get through. This is true when:
 \begin{equation}
	\begin{array}{lll}
		2\cdot \alpha\cdot(1-\alpha) & > & (1-\alpha)^2 \\
		\implies \alpha & > & \frac{1}{3}
		
	\end{array}
 \end{equation}  \newline
 However this behavior doesn't scale arbitrarily and whenever $\alpha >  \frac{1}{3}$ this behavior increases the security of the protocol. \newline 
\end{proof}

\newpage
\section{Proof of Lemma \ref{lem:1_2}}

\lemmab*

 \begin{proof}
 In Phase 1, $W$ is only required to store the message $m$ that is last signed by him at this point. The only way to deviate from the protocol would be to forge a message $\tilde{m}$ s.t. $d_m > d_{\tilde{m}}$. We show that $W$ won't be able to construct such a message.
 	 If $W$ doesn't control any nodes in the tree which have a shorter path to the root, then he can not create a valid  $\tilde{m}$, due to the layered construction of update messages.
 	If $W$ does control a watchtower in the tree which is closer to the root, receiving messages after $i$ hops. Then $m$ is in $S_{W,i}$ which means that we can apply lemma \ref{lem:1_1}, as creating $\tilde{m}$ would be equivalent to creating more than $N$ messages at level $i$.
 \end{proof}

\section{Proof of Lemma \ref{lem:1_3}}

\lemmac*
 \begin{proof}
  $W$ can take three different actions to try increasing the probability that his message will be included in the next block. We show that none of them in fact increase that probability.
  	Firstly, $W$ could send $m \in S_{W,i}$ too early, e.g. in round $j<i$. In this case the transaction is invalid, the network will reject the message and might disconnect from $W$.
  	Secondly, $W$ could send $m$ in a later round. It is easy to see that there is no advantage in doing so.
  	Lastly, $W$ could send $\tilde{m}$ which belongs to an older update transaction than $m_W$ but a newer one than settlement transaction. If $\tilde{m}$ is in fact included in the blockchain, then the protocol terminates successfully nevertheless. However, there is no incentive for $W$ to store $\tilde{m}$ as storing $m$ instead is a dominant strategy, since the set of potential settlement transactions for which $m$ is a valid proof-of-fraud is a subset of the set of settlement transactions for which $\tilde{m}$ is a valid proof-of-fraud.
  \end{proof}


\end{document}